\documentclass[aps,pre,reprint,amsmath,amssymb,floatfix,superscriptaddress,nofootinbib]{revtex4-2}

\usepackage{graphicx}
\usepackage{dcolumn}
\usepackage{bm}
\usepackage{hyperref}
\usepackage{mathtools}
\usepackage{amsthm}
\usepackage{xcolor}
\newcommand{\TN}{\mathbb{T}_N}
\newcommand{\ONK}{\Omega_{N,K}}
\newcommand{\E}{\mathbb{E}}
\newcommand{\ONKa}{\Omega_{N,K}^{\mathrm{a}}}
\newcommand{\Cov}{\mathrm{Cov}}

\newtheorem{theorem}{Theorem}

\newtheorem{remark}[theorem]{Remark}

\begin{document}

\title{Tandem Exclusion Process}

\author{Ngo Phuoc Nguyen Ngoc}
\affiliation{Institute of Research and Development, Duy Tan University, Da Nang, 550000, Viet Nam}
\affiliation{Faculty of Natural Sciences, Duy Tan University, Da Nang, 550000, Viet Nam}

\author{Lam Thi Nhung}
\affiliation{Faculty of Information Technology, Ho Chi Minh City University of Economics and Finance, Ho Chi Minh, 700000, Viet Nam}

\author{Huynh Anh Thi}
\affiliation{Institute of Research and Development, Duy Tan University, Da Nang, 550000, Viet Nam}
\affiliation{Faculty of Natural Sciences, Duy Tan University, Da Nang, 550000, Viet Nam}

\date{\today}

\begin{abstract}
We introduce the \emph{tandem exclusion process} (TEP), a one-dimensional stochastic lattice model motivated by tandem running in ants. Particles evolve through two cooperative local transitions, $110\to101$ at rate $\alpha$ (leader advancement) and $101\to011$ at rate $\beta$ (follower recovery). We prove that the stationary measure on the dynamically active sector is the Gibbs measure $\pi\propto q^{H(\eta)}$, where $q=\beta/\alpha$ and $H(\eta)$ counts neighboring occupied pairs, and derive exact closed-form expressions for the stationary current and spatial correlations using transfer-matrix methods. The current is asymmetric under particle--hole exchange $\rho\mapsto1-\rho$, with its maximum occurring at densities strictly larger than $1/2$, in contrast to the symmetric current $\rho(1-\rho)$ of the totally asymmetric simple exclusion process (TASEP). For $q>1$, cooperative dynamics enhances the current above the TASEP value and generates strong spatial clustering; in the limit $q\to\infty$, the current approaches $J\to\alpha\rho$, corresponding to nearly unconstrained collective transport. These results suggest that tandem coordination alone can substantially enhance collective transport efficiency at moderate and high densities, even without pheromone-mediated long-range communication.
\end{abstract}

\maketitle

% =========================================================
\section{Introduction}
\label{sec:intro}
% =========================================================

Collective transport in biological and physical systems often emerges from local cooperative interactions rather than from independent particle motion~\cite{Derrida1998,Liggett1985,Spitzer1970}. Examples range from intracellular cargo transport and vehicular traffic to pedestrian dynamics and coordinated motion in social insects~\cite{Bonabeau1999,Camazine2001,Chowdhury2000,Helbing2001,Schadschneider2010}. A central problem in nonequilibrium statistical mechanics is to understand how simple microscopic interaction rules generate collective macroscopic transport phenomena, and in particular how local cooperation can qualitatively alter transport efficiency and spatial organization~\cite{Derrida1998,Liggett1985,Chowdhury2000,Schadschneider2010}.

Among social insects, ant colonies exhibit several distinct mechanisms of cooperative recruitment. The most extensively studied is pheromone-mediated trail formation, in which ants communicate indirectly through chemical signals deposited on the substrate~\cite{Bonabeau1999,Beckers1992,Chowdhury2002,Deneubourg1989,John2009,Ngoc2024}. Through such mechanisms, colonies achieve large-scale self-organized transport without centralized coordination. Collective insect behavior has also inspired algorithms and models in optimization, robotics, and distributed systems~\cite{Anderson1999,Bonabeau2000,Dorigo1999,Krieger2000}.

Tandem running is a qualitatively different form of cooperative recruitment. In this behavior, an informed ant guides a naive nestmate toward a food source or a new nest site through repeated direct physical contact~\cite{FranksRichardson2006,Moeglich1974}. Unlike pheromone-based recruitment, tandem running relies on direct pairwise coordination: leader and follower must remain in near contact, and the motion of each individual depends on the instantaneous relative position of its partner~\cite{Richardson2007}. Experimental observations indicate that tandem motion proceeds through an alternating two-stage mechanism: the leader first advances, temporarily opening a gap, after which the follower catches up and restores the tandem configuration before the next advance~\cite{Richardson2007}. This repeated cycle of leader initiative and follower recovery is intrinsically cooperative and has no close analogue in classical exclusion models of driven transport.

From the viewpoint of interacting particle systems, tandem running raises a natural question: how does cooperative pairwise motion modify the collective transport properties of exclusion-driven systems?

Most classical driven lattice gases are based on independent particle motion. In the asymmetric simple exclusion process (ASEP) and its totally asymmetric variant (TASEP), particles hop whenever the target site is vacant~\cite{Derrida1998,Liggett1985,Spitzer1970,Derrida1993}. The resulting stationary states are product measures, the current is symmetric under particle-hole exchange, and the maximal transport occurs at half filling. These features reflect the absence of cooperative mechanisms beyond hard-core exclusion.

In this work, we introduce the \emph{tandem exclusion process} (TEP), a minimal one-dimensional driven lattice gas on the discrete torus in which transport occurs only through coordinated pair motion. The dynamics consists of two local transitions: a head move
\[
110 \to 101
\]
at rate $\alpha$, in which the leading particle advances, and a tail move
\[
101 \to 011
\]
at rate $\beta$, in which the trailing particle catches up and restores the tandem configuration. As a consequence, isolated particles are immobile and transport can occur only through coordinated leader--follower motion.

The TEP is related to several classes of cooperative driven lattice gases studied in nonequilibrium statistical mechanics. Nearest-neighbor Gibbs stationary states also arise in exclusion processes with short-range interactions, such as the Katz--Lebowitz--Spohn (KLS) model~\cite{KLS1984}, asymmetric exclusion processes with next-nearest-neighbor interactions~\cite{Antal2000}, and other cooperative exclusion models with local interaction-driven transport~\cite{Ngoc2025}. Related kinetic constraints also appear in facilitated exclusion processes~\cite{Baik2019,Gabel2010,Gabel2011,Zhao2019}, where particle motion is enabled only in the presence of nearby occupied sites and may lead to jamming phenomena and nontrivial current-density relations.

In contrast to facilitated exclusion processes, where transport is still associated with individual particle hops, the TEP supports transport through cooperative tandem structures. This distinction leads to qualitatively different stationary currents and clustering behavior.

Despite the simplicity of the local rules, the model combines irreversible tandem dynamics with an exactly solvable nearest-neighbor Gibbs stationary structure. The continuous-time dynamics admits explicit formulas for the stationary current and spatial correlations, and exhibits asymmetric transport, strong clustering, and a correlation length growing as $q^{1/2}$ in the strong-recovery regime $q=\beta/\alpha\gg1$.

The main results of this paper can be summarized as follows. We prove that the unique stationary measure on the dynamically active sector is a nearest-neighbor Gibbs measure
\[
\pi(\eta)\propto q^{H(\eta)},
\]
where $H(\eta)$ counts neighboring occupied pairs (Theorem~\ref{thm:invariant}). Using transfer-matrix methods, we derive a closed-form expression for the stationary current at all densities and interaction strengths (Theorem~\ref{thm:current}), and characterize the spatial organization of the stationary state through exact pair correlations and a correlation length.

The transport behavior differs qualitatively from classical TASEP in several respects. First, the current is not symmetric under the particle-hole transformation $\rho\mapsto1-\rho$. Second, the maximal current occurs at densities strictly larger than $1/2$. Third, for $q>1$, cooperative clustering enhances the current above the TASEP value and produces increasingly extended cooperative tandem motion. Remarkably, even at $q=1$ -- where the stationary state reduces to the Bernoulli product measure -- the current remains asymmetric, showing that kinetic constraints alone are sufficient to break particle--hole symmetry without any stationary interaction.

The remainder of the paper is organized as follows. Section~\ref{sec:model} introduces the model and stationary measure. Section~\ref{sec:current} derives the stationary current and clustering properties. Appendices contain transfer-matrix computations.

% =========================================================
\section{Model and stationary measure}
\label{sec:model}
% =========================================================

We consider the tandem exclusion process (TEP) on the discrete one-dimensional torus
$
\TN=\mathbb{Z}/N\mathbb{Z}.
$
A configuration is denoted by
$
\eta=(\eta_x)_{x\in\TN}\in\{0,1\}^{\TN},
$
where $\eta_x=1$ if site $x$ is occupied and $\eta_x=0$ otherwise.

The dynamics consists of the two local transitions
\begin{align}
110 &\xrightarrow{\alpha} 101, \label{eq:headmove}\\
101 &\xrightarrow{\beta} 011, \label{eq:tailmove}
\end{align}
while isolated particles remain immobile. The transition~\eqref{eq:headmove} describes the advancement of the leading particle, whereas~\eqref{eq:tailmove} corresponds to the recovery motion of the trailing particle.

The total particle number
$
K(\eta)=\sum_{x\in\TN}\eta_x
$
is conserved, and the process evolves on the sector
\[
\ONK = \Bigl\{ \eta\in\{0,1\}^{\TN}: \sum_{x\in\TN}\eta_x=K \Bigr\}.
\]
The process is governed by the generator
\[
L_N=L_N^{\mathrm h}+L_N^{\mathrm t},
\]
where
\begin{align*}
L_N^{\mathrm h}f(\eta) &= \sum_{x\in\TN} \alpha\, \eta_{x-1}\eta_x(1-\eta_{x+1}) \bigl[ f(\eta^{x,x+1})-f(\eta) \bigr], \label{eq:Lh}\\
L_N^{\mathrm t}f(\eta) &= \sum_{x\in\TN} \beta\, \eta_x(1-\eta_{x+1})\eta_{x+2} \bigl[ f(\eta^{x,x+1})-f(\eta) \bigr], %\label{eq:Lt}
\end{align*}
and $\eta^{x,x+1}$ denotes the configuration obtained by moving a particle from site $x$ to site $x+1$.

We now characterize the stationary state of the process. Let
\begin{equation}\label{eq:Ham}
H(\eta)=\sum_{x\in\TN}\eta_x\eta_{x+1}
\end{equation}
denote the number of neighboring occupied pairs, and define
$
q=\beta/\alpha.
$
We introduce the Gibbs weight
\begin{equation}\label{eq:pi}
\pi_{N,K}(\eta) = \frac{q^{H(\eta)}}{Z_{N,K}}, \qquad Z_{N,K} = \sum_{\xi\in\ONK}q^{H(\xi)}.
\end{equation}

Not all configurations are dynamically active. Define
\begin{equation*}\label{eq:active}
	\ONKa
	=
	\bigl\{
	\eta\in\ONK :
	\eta \text{ contains a pattern }110\text{ or }101
	\bigr\}.
\end{equation*}
A configuration $\eta\in\ONK\setminus\ONKa$ contains neither pattern, so no head or tail move is enabled at $\eta$; we call such $\eta$ \emph{frozen}. Frozen configurations are absorbing states of the dynamics.

\begin{remark}[Frozen configurations]\label{rem:frozen}
A configuration $\eta\in\ONK$ is frozen if and only if no two occupied sites are within circular distance $2$, i.e.\ every pair of consecutive particles is separated by at least two empty sites. By a counting argument, frozen configurations exist in $\ONK$ if and only if $3K\le N$. For $3K > N$, $\ONKa = \ONK$ and the conditioning in~\eqref{eq:pi-active} is vacuous; for $3K \le N$, both $\ONKa$ and the set of frozen configurations are nonempty, and the dynamics decomposes into the irreducible class $\ONKa$ together with the absorbing frozen states.
\end{remark}

\begin{theorem}[Stationary measure]\label{thm:invariant}
Let $N\ge3$, $1\le K\le N$, and $\alpha,\beta>0$. The dynamics restricted to $\ONKa$ is irreducible, and the conditional Gibbs measure
\begin{equation}\label{eq:pi-active}
\pi_{N,K}(\eta) = \frac{q^{H(\eta)}}{Z_{N,K}^{\mathrm a}} \mathbf{1}_{\eta\in\ONKa}, \quad Z_{N,K}^{\mathrm a} = \sum_{\xi\in\ONKa}q^{H(\xi)},
\end{equation}
is the unique stationary probability measure on $\ONKa$.
\end{theorem}

\begin{proof}
The full state space $\ONK$ is not irreducible in general. Indeed, configurations containing neither $110$ nor $101$ have no available transition and are therefore frozen absorbing states. Hence uniqueness of the stationary measure can only be expected after restricting the dynamics to the active sector $\ONKa$.

We first verify stationarity of the Gibbs weight on the active sector. It is enough to prove the global balance identity
\begin{equation}\label{eq:gb}
\sum_{\eta'\ne\eta}\pi_{N,K}(\eta')\,c(\eta'\to\eta) = \pi_{N,K}(\eta) \sum_{\eta'\ne\eta}c(\eta\to\eta')
\end{equation}
for every $\eta\in\ONKa$, where $c(\cdot\to\cdot)$ denotes the transition rates. For frozen configurations the identity is trivial, since both the total inflow and the total outflow vanish.

For a configuration $\eta$, introduce the local indicators
\begin{equation}\label{eq:indicators}
\begin{aligned}
A_x(\eta) &= \eta_x(1-\eta_{x+1})\eta_{x+2}, &&\text{for the pattern $101$,}\\
B_x(\eta) &= (1-\eta_x)\eta_{x+1}\eta_{x+2}, &&\text{for the pattern $011$,}\\
C_x(\eta) &= \eta_x\eta_{x+1}(1-\eta_{x+2}), &&\text{for the pattern $110$.}
\end{aligned}
\end{equation}
The total rate out of $\eta$ is
\[
\sum_{\eta'\ne\eta}c(\eta\to\eta') = \sum_{x\in\TN} \bigl[ \alpha C_x(\eta)+\beta A_x(\eta) \bigr],
\]
because head moves occur at $110$ patterns and tail moves occur at $101$ patterns.

We now compute the total inflow into $\eta$. A head move can produce $\eta$ only when $\eta$ has a $101$ pattern; the corresponding predecessor has $110$ at the same triple. Similarly, a tail move can produce $\eta$ only when $\eta$ has a $011$ pattern; the corresponding predecessor has $101$. Denote these predecessors by $\eta^{[h,x]}$ and $\eta^{[t,x]}$, respectively. Then
\begin{align*}
	&\ \ \sum_{\eta'\ne\eta}\pi_{N,K}(\eta')c(\eta'\to\eta) \\
	& = \sum_{x\in\TN} \left[ \alpha\,\pi_{N,K}(\eta^{[h,x]})A_x(\eta) + \beta\,\pi_{N,K}(\eta^{[t,x]})B_x(\eta) \right].
\end{align*}

The only nontrivial point is to compare the Gibbs weights of $\eta$ and its local predecessors. A direct count of the changed nearest-neighbor bonds gives
\[
H(\eta^{[h,x]})-H(\eta)=1-\eta_{x+3}, \quad H(\eta^{[t,x]})-H(\eta)=\eta_{x-1}-1.
\]
Since $\pi_{N,K}(\eta)\propto q^{H(\eta)}$, it follows that
\begin{equation}\label{eq:pi-ratio}
\frac{\pi_{N,K}(\eta^{[h,x]})}{\pi_{N,K}(\eta)} = q^{1-\eta_{x+3}}, \quad \frac{\pi_{N,K}(\eta^{[t,x]})}{\pi_{N,K}(\eta)} = q^{\eta_{x-1}-1}.
\end{equation}

After dividing~\eqref{eq:gb} by $\pi_{N,K}(\eta)$ and substituting the identities~\eqref{eq:pi-ratio}, it remains to prove
\[
\sum_{x\in\TN}G_x(\eta)=0,
\]
where
\begin{equation*}\label{eq:Gx-def}
G_x(\eta) = \alpha q^{1-\eta_{x+3}}A_x(\eta) + \beta q^{\eta_{x-1}-1}B_x(\eta) - \alpha C_x(\eta) - \beta A_x(\eta).
\end{equation*}
Using $\eta_y\in\{0,1\}$ and
\[
q^{\eta_y}=1+(q-1)\eta_y,
\]
together with $q=\beta/\alpha$, one obtains
\[
\alpha q^{1-\eta_{x+3}}A_x = \beta A_x+(\alpha-\beta)\eta_{x+3}A_x,
\]
and
\[
\beta q^{\eta_{x-1}-1}B_x = \alpha B_x+(\beta-\alpha)\eta_{x-1}B_x.
\]
Hence
\begin{equation*}\label{eq:Gx-reduced}
G_x(\eta) = \alpha\bigl(B_x(\eta)-C_x(\eta)\bigr) + (\alpha-\beta) \bigl( \eta_{x+3}A_x(\eta)-\eta_{x-1}B_x(\eta) \bigr).
\end{equation*}

The sum of the first term vanishes by translation invariance on the torus. Indeed,
\[
B_x(\eta)-C_x(\eta) = \eta_{x+1}\eta_{x+2} - \eta_x\eta_{x+1},
\]
and therefore
\[
\sum_{x\in\TN}\bigl(B_x(\eta)-C_x(\eta)\bigr)=0.
\]
The second term also sums to zero. Indeed,
\[
\eta_{x+3}A_x(\eta) = \eta_x(1-\eta_{x+1})\eta_{x+2}\eta_{x+3},
\]
whereas
\[
\eta_{x-1}B_x(\eta) = \eta_{x-1}(1-\eta_x)\eta_{x+1}\eta_{x+2}.
\]
Both sums count the same four-site pattern $1011$ on the torus, only with a shifted index. Hence
\[
\sum_{x\in\TN} \eta_{x+3}A_x(\eta) = \sum_{x\in\TN} \eta_{x-1}B_x(\eta).
\]
Combining the two cancellations gives
\[
\sum_{x\in\TN}G_x(\eta)=0,
\]
and therefore the Gibbs measure is stationary.

We finally prove irreducibility on $\ONKa$. Fix a reference configuration
\[
\eta^\star = (1,1,\ldots,1,0,\ldots,0),
\]
consisting of one contiguous block of $K$ particles.

Starting from any active configuration $\eta\in\ONKa$, there exists at least one movable tandem structure. By alternating head and tail moves, this tandem pair can be translated along the torus. Moreover, whenever a particle is separated from the main cluster by a finite gap, the tandem mechanism allows the gap to be reduced by one site after finitely many updates. Repeating this procedure successively, all particles can be merged into a single contiguous cluster, yielding $\eta^\star$.

The same argument applied in reverse shows that every active configuration can be reached from $\eta^\star$. Hence any two configurations in $\ONKa$ communicate, proving that the dynamics restricted to $\ONKa$ is irreducible.

Because $\ONKa$ is finite and the restricted chain is irreducible, the stationary probability measure on $\ONKa$ is unique. This completes the proof.
\end{proof}

\begin{remark}\label{rem:gibbs}
Writing
$
q=e^{-\kappa J},
$
where $\kappa$ plays the role of an inverse temperature and $J$ an effective nearest-neighbor interaction strength, the stationary measure~\eqref{eq:pi-active} takes the Gibbs form
\[
\pi(\eta) \propto \exp\!\bigl(-\kappa J\,H(\eta)\bigr), \qquad H(\eta)=\sum_x\eta_x\eta_{x+1}.
\]
Thus, at the level of the stationary state, the TEP behaves formally like a one-dimensional equilibrium system with nearest-neighbor interaction, even though the underlying dynamics is irreversible and out of equilibrium.

When $J<0$ (equivalently $q>1$), neighboring occupied sites are favored, corresponding to an attractive interaction and leading to clustering of tandem groups. Conversely, when $J>0$ (equivalently $q<1$), adjacent occupied sites are penalized, corresponding to a repulsive interaction and favoring spatially dispersed configurations. At the noninteracting point $J=0$ (that is, $q=1$), the interaction disappears and the stationary measure reduces to the uniform measure on $\ONKa$, corresponding in the thermodynamic limit to the Bernoulli product measure.
\end{remark}

% =========================================================
\section{Current and clustering}
\label{sec:current}
% =========================================================

\subsection{Stationary current}

The central nonequilibrium observable of the tandem exclusion process is the stationary particle current. Let $\E_\pi$ denote expectation with respect to the stationary measure $\pi_{N,K}$~\eqref{eq:pi-active}. Since $\pi_{N,K}$ is translation invariant, the mean current across any bond is independent of the bond location and is obtained by averaging the local transition rates:
\begin{equation}\label{eq:J-as-3pt}
J = \alpha\,\E_\pi[\eta_1\eta_2(1-\eta_3)] + \beta\,\E_\pi[\eta_1(1-\eta_2)\eta_3].
\end{equation}
The two terms correspond respectively to the head move $110\to101$ (leader advancement) and the tail move $101\to011$ (follower recovery). The three-point correlation functions appearing in~\eqref{eq:J-as-3pt} are evaluated by the grand-canonical transfer-matrix method, with details deferred to Appendix~\ref{sec:transfer}.

\begin{theorem}[Stationary current]\label{thm:current}
In the thermodynamic limit at density $\rho\in(0,1)$, the stationary current is
\begin{equation}\label{eq:J-formula}
J(\rho;\alpha,\beta) = \frac{\alpha q\rho s(s+2+qs)}{(s+1)(1+qs)^2},
\end{equation}
where $q=\beta/\alpha$ and $s=s(\rho,q)$ is the unique positive solution of
\begin{equation}\label{eq:s-eq}
q(1-\rho)s^2+(1-2\rho)s-\rho=0,
\end{equation}
given explicitly by
\begin{equation}\label{eq:s-formula}
s(\rho,q) = \frac{-(1-2\rho) + \sqrt{(1-2\rho)^2+4q(1-\rho)\rho}}{2q(1-\rho)}.
\end{equation}
\end{theorem}

\begin{remark}[Comparison with TASEP]
For the TASEP,
\[
J_{\mathrm{TASEP}}(\rho)=\rho(1-\rho),
\]
and the maximal current occurs at $\rho=1/2$ due to particle-hole symmetry. In contrast, the TEP current is intrinsically asymmetric because transport requires cooperative tandem structures. As a result, the optimal transport regime is shifted toward higher densities where neighboring particles are more abundant.
\end{remark}

\begin{figure}[t]
\centering
\includegraphics[width=0.48\textwidth]{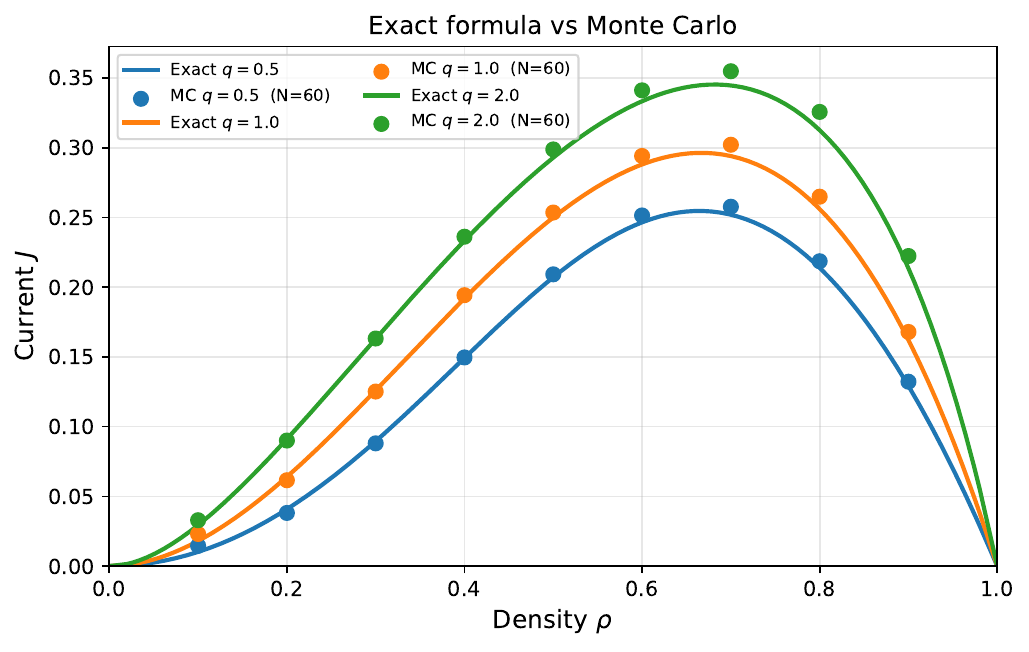}
\caption{Stationary current $J(\rho;\alpha,\beta)$ for $q\in\{0.5,1,2\}$ with $\alpha=1$. Solid curves: exact formula~\eqref{eq:J-formula}. Filled circles: Monte Carlo simulations ($N=60$, $5\times10^5$ updates).}
\label{fig:tepsimcomparison}
\end{figure}

\begin{figure}[t]
\centering
\includegraphics[width=0.48\textwidth]{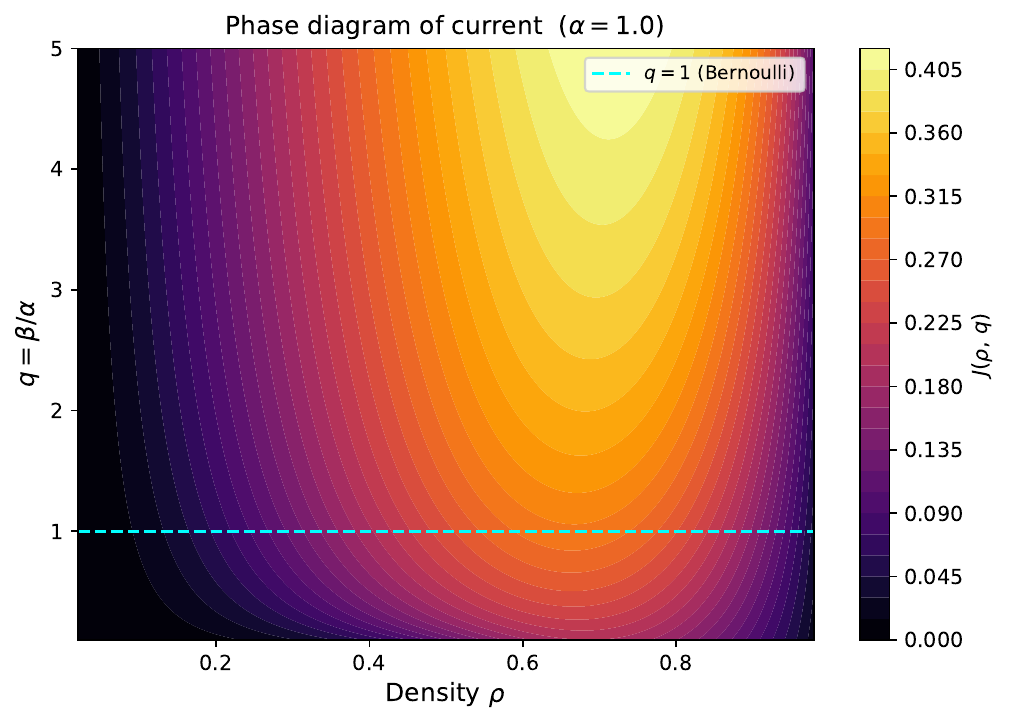}
\caption{Phase diagram of the stationary current $J(\rho,q)$ with $\alpha=1$. Colors encode the current magnitude from low (dark) to high (bright); the dashed line $q=1$ marks the Bernoulli locus.}
\label{fig:tepphasediagram}
\end{figure}

Figure~\ref{fig:tepsimcomparison} compares formula~\eqref{eq:J-formula} with Monte Carlo simulations over the full density range, showing excellent agreement at every value of $q$, and Figure~\ref{fig:tepphasediagram} displays the corresponding phase diagram in the $(\rho,q)$-plane. The current exhibits three structural features, all originating from the cooperative nature of the dynamics.

\medskip
\noindent\textbf{(i) Broken particle--hole symmetry.}
Because transport requires the simultaneous presence of a leader and a follower, the current $J(\rho;\alpha,\beta)$ is not invariant under $\rho\mapsto1-\rho$, and its maximum is attained at a density $\rho^*>1/2$. This asymmetry is evident in Figure~\ref{fig:tepsimcomparison}: the curves are skewed toward larger densities, with maxima located near $\rho\approx2/3$ rather than at half filling. Remarkably, the asymmetry persists even at the noninteracting point $q=1$, where~\eqref{eq:s-eq} yields $s=\rho/(1-\rho)$ and
\begin{equation}\label{eq:J-bernoulli}
J(\rho;\alpha,\alpha) = 2\alpha\rho^2(1-\rho),
\end{equation}
which is maximized at exactly $\rho=2/3$. Although the stationary state at $q=1$ reduces to the Bernoulli product measure, the kinetic constraint -- two adjacent particles are required to enable a transition -- is by itself sufficient to break particle--hole symmetry, in the absence of any stationary interaction.

\medskip
\noindent\textbf{(ii) Monotone enhancement and the strong-recovery limit.}
For every fixed $\rho\in(0,1)$, the current $J(\rho;\alpha,\beta)$ is a strictly increasing function of $q$: faster follower recovery stabilizes tandem pairs and accelerates transport. This monotonicity is clearly visible in Figure~\ref{fig:tepsimcomparison}, where the three curves corresponding to $q\in\{0.5,1,2\}$ are uniformly ordered, and in Figure~\ref{fig:tepphasediagram}, where the current magnitude increases along every vertical line.

Indeed, differentiating~\eqref{eq:s-eq} with respect to $q$ at fixed $\rho$ gives
\[
\frac{\partial s}{\partial q} = -\frac{(1-\rho)s^2}{2q(1-\rho)s+(1-2\rho)},
\]
and the closed form~\eqref{eq:s-formula} yields the bound $s(1-2\rho)\le\rho$, so the denominator equals $[2\rho-s(1-2\rho)]/s>0$, whence $\partial s/\partial q<0$. On the other hand, eliminating $q$ from~\eqref{eq:J-formula} by means of~\eqref{eq:s-eq} gives the equivalent expression
\[
J(\rho,q) = \frac{\rho}{s+1} - \frac{(1-2\rho)(1-\rho)s^2}{\rho(s+1)^2},
\]
whence
\[
\frac{dJ}{ds} = -\frac{\rho^2+s(5\rho^2-6\rho+2)}{\rho(s+1)^3}.
\]
The discriminant of $5\rho^2-6\rho+2$ is $-4<0$, so the numerator is strictly positive for all $\rho\in(0,1)$ and $s>0$, giving $dJ/ds<0$. Combining the two signs yields $\partial J/\partial q>0$.

In the strong-recovery limit $q\to\infty$, tail moves become effectively instantaneous: once a leader produces a $101$ configuration, the follower rapidly restores the tandem pair before further rearrangements occur. As a result, tandem clusters propagate coherently through the available vacancies, and the suppressing influence of exclusion on the current becomes much weaker than in the TASEP.

From~\eqref{eq:s-formula},
\[
s\sim\sqrt{\rho/(q(1-\rho))}\to0,
\]
and substitution into~\eqref{eq:J-formula} gives
\begin{equation}\label{eq:J-largeq}
J(\rho;\alpha,\beta) \longrightarrow \alpha\rho, \qquad q\to\infty.
\end{equation}
Thus, in the strong-recovery regime, cooperative tandem clusters move with an effective velocity approaching $\alpha$, leading to an almost linear current--density relation.

The brightening of the upper region in Figure~\ref{fig:tepphasediagram} as $q$ increases reflects this approach toward the highly cooperative transport regime $J(\rho)\approx\alpha\rho$.

\medskip
\noindent\textbf{(iii) Suppression in the slow-recovery regime.}
For $q<1$, tandem structures are unstable: the $101$ pattern produced by a head move is more likely to dissolve through a subsequent head move at a neighboring triple than to be repaired by a tail move. Cooperative propagation is therefore strongly suppressed and the current decreases below the Bernoulli value~\eqref{eq:J-bernoulli}, as illustrated by the lowest curve ($q=0.5$) in Figure~\ref{fig:tepsimcomparison} and by the dark band below the $q=1$ line in Figure~\ref{fig:tepphasediagram}.

% =========================================================
\subsection{Clustering structure and correlation length}
\label{sec:clustering}
% =========================================================

The transport features identified in the previous subsection are the macroscopic manifestation of an underlying spatial reorganization of the stationary state. This reorganization is most directly visible in the spacetime diagrams of Figure~\ref{fig:spacetime}, which show typical trajectories at density $\rho=1/2$ for three representative values of $q$. For $q=0.05$ (slow-recovery regime), particles trace sparse, essentially independent trajectories with no large-scale coherence, reflecting the suppressed current of regime~(iii). At $q=1$, the configurations are Bernoulli distributed and the spacetime pattern displays no characteristic structure: no visible clusters, no typical inter-particle spacing. For $q=20$ (strong-recovery regime), the dynamics organizes into broad diagonal bands of characteristic width, corresponding to tandem clusters that propagate coherently as nearly rigid moving structures separated by persistent empty gaps.

\begin{figure*}[t]
\centering
\includegraphics[width=\textwidth]{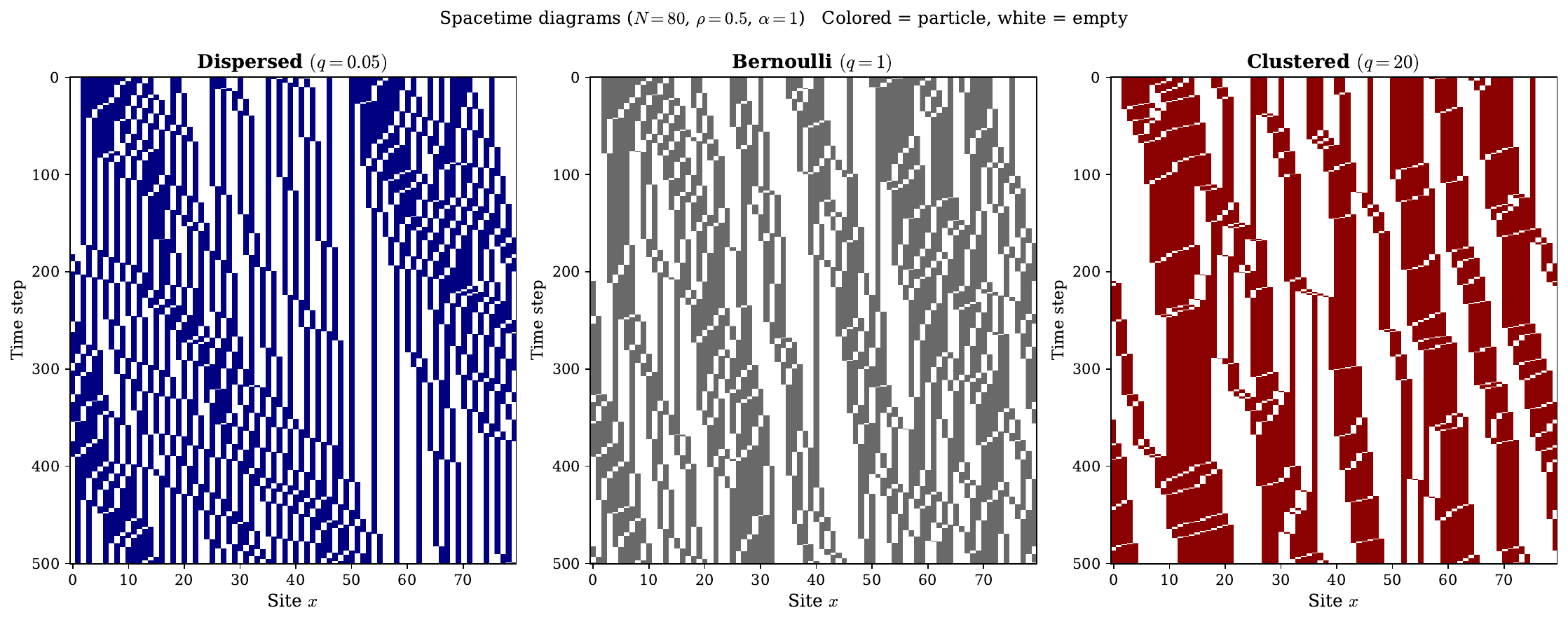}
\caption{Spacetime diagrams ($N=80$, $\rho=0.5$, $\alpha=1$). Left: dispersed regime ($q=0.05$). Center: Bernoulli regime ($q=1$). Right: clustered regime ($q=20$), where coherent moving tandem groups appear as broad diagonal bands.}
\label{fig:spacetime}
\end{figure*}

We now quantify these observations at two complementary scales: the nearest-neighbor pair correlation, which measures local clustering, and the asymptotic correlation length, which sets the typical extent of coherent tandem structures.

\medskip
\noindent\textbf{Nearest-neighbor clustering.}
A natural local indicator of clustering is the pair correlation ratio
\begin{equation}\label{eq:g}
g(\rho,q) := \frac{\E_\pi[\eta_0\eta_1]}{\rho^2},
\end{equation}
which compares the probability of adjacent occupation with the Bernoulli baseline $\rho^2$. The transfer-matrix computation (Appendix~\ref{sec:transfer},~\eqref{eq:E11-app}) gives
\begin{equation}\label{eq:E11}
\E_\pi[\eta_0\eta_1] = \frac{q\rho s}{1+qs},
\end{equation}
from which $g>1$, $g=1$, and $g<1$ correspond to clustered, uncorrelated, and dispersed regimes, respectively. The crossover between clustered and dispersed regimes occurs precisely at $q=1$ for every density.

The two extreme regimes $q\to\infty$ and $q\to0$ exhibit qualitatively different spatial organization. For $q\to\infty$, follower recovery becomes effectively instantaneous, so particles aggregate into coherent tandem clusters. Consequently,
\[
\E_\pi[\eta_0\eta_1]\to\rho,
\]
meaning that almost every occupied site has an occupied neighbor, and
\[
g(\rho,q)\to\frac1\rho,
\]
the largest value consistent with $\E_\pi[\eta_0\eta_1]\le\rho$ (this inequality comes from~\eqref{eq:E11}).

\begin{figure*}[t]
\centering
\includegraphics[width=\textwidth]{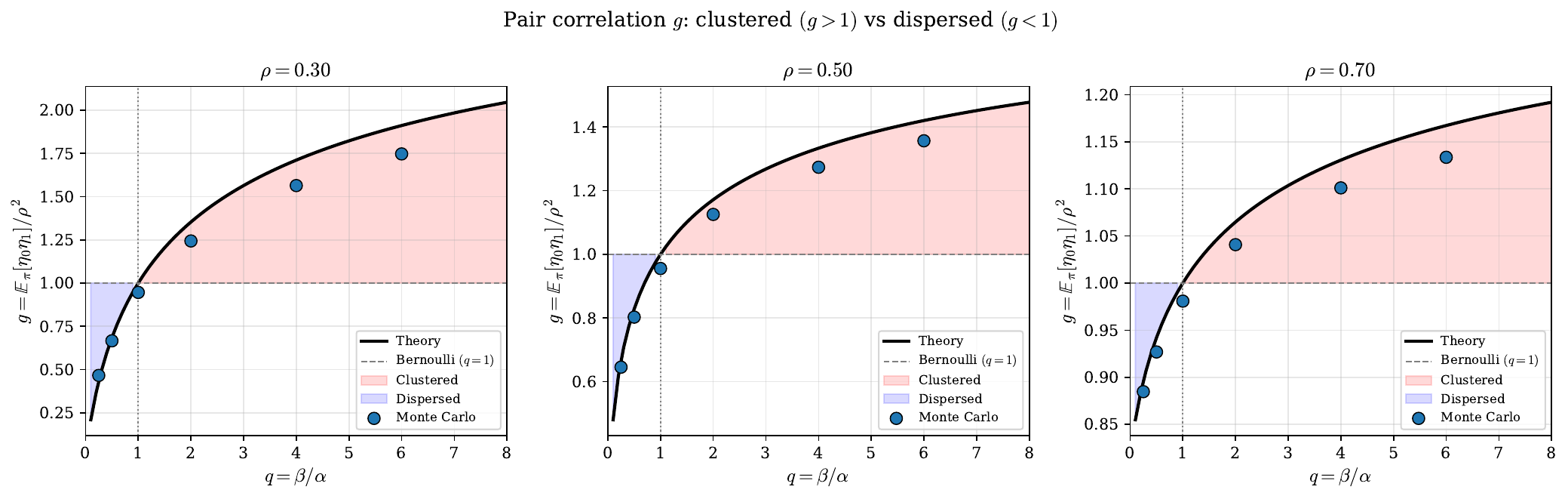}
\caption{Pair correlation ratio $g=\E_\pi[\eta_0\eta_1]/\rho^2$ as a function of $q=\beta/\alpha$ for several densities. Solid curves: exact formula~\eqref{eq:E11}. Filled circles: Monte Carlo simulations ($N=50$, $8000$ samples). The crossover between clustered ($g>1$) and dispersed ($g<1$) regimes occurs at $q=1$.}
\label{fig:clustering-g}
\end{figure*}

\begin{figure}[t]
\centering
\includegraphics[width=0.48\textwidth]{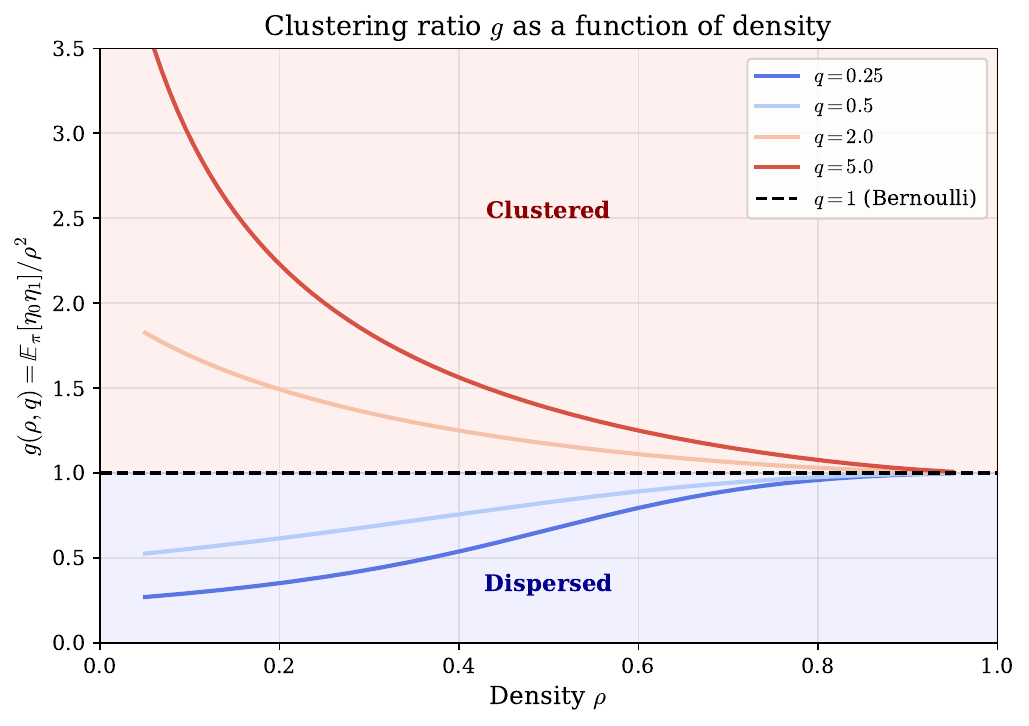}
\caption{Clustering ratio $g(\rho,q)$ as a function of density. For $q>1$, clustering is strongest at low density, where pairs are generated cooperatively rather than by packing constraints. All curves converge to $g=1$ as $\rho\to1$.}
\label{fig:clustering-rho}
\end{figure}

In contrast, when $q\to0$, tandem structures become unstable and adjacent occupied pairs are strongly suppressed. From~\eqref{eq:s-formula},
\[
qs \longrightarrow \frac{-(1-2\rho)+|1-2\rho|}{2(1-\rho)}, \qquad q\to0.
\]
Hence, for $\rho<1/2$,
\[
\E_\pi[\eta_0\eta_1]\to0, \qquad g(\rho,q)\to0,
\]
corresponding to asymptotically dispersed configurations. For $\rho>1/2$,
\[
\E_\pi[\eta_0\eta_1]\to2\rho-1,
\]
showing that complete dispersion becomes impossible at high density for geometric reasons. At the special density $\rho=1/2$, the situation is qualitatively different. Since
\[
\pi(\eta)\propto q^{H(\eta)}, \qquad H(\eta)=\sum_x\eta_x\eta_{x+1},
\]
the limit $q\to0$ concentrates the stationary measure on configurations minimizing $H(\eta)$, i.e.\ configurations without adjacent occupied pairs. At density $\rho=1/2$, the only such configurations are the two alternating states
\[
101010\cdots, \qquad 010101\cdots.
\]
The system therefore develops long-range staggered order, which is reflected in the divergence of the correlation length discussed below.

Figures~\ref{fig:clustering-g} and~\ref{fig:clustering-rho} illustrate these behaviors. The clustering ratio $g$ varies most strongly with $q$ at low density: in Figure~\ref{fig:clustering-rho} the family of curves spreads from $g\approx 0.25$ to $g\approx 3.5$ at $\rho=0.1$, but collapses toward a single value as $\rho$ increases. The reason is twofold. At low density, the Bernoulli baseline $\rho^2$ is small, so deviations from independence are strongly amplified in the ratio $g=\E_\pi[\eta_0\eta_1]/\rho^2$. At high density, geometric packing forces neighboring occupied pairs, and consequently $g(\rho,q)\to1$ as $\rho\to1$, independently of $q$.

\medskip
\noindent\textbf{Correlation length.}
Beyond nearest-neighbor correlations, the spatial structure of $\pi_{N,K}$ is governed by the spectral gap of the transfer matrix $T$~\eqref{eq:T-matrix}. Standard transfer-matrix arguments imply that connected correlations decay exponentially with distance. Indeed, since the transfer matrix has two eigenvalues $\lambda_+>|\lambda_-|$, the large-distance behavior is dominated by the spectral ratio $\lambda_-/\lambda_+$, yielding
\begin{equation*}\label{eq:cov-decay}
\Cov_\pi(\eta_0,\eta_r) = A(\rho,q) \left(\frac{\lambda_-}{\lambda_+}\right)^r + o\!\left(\left|\frac{\lambda_-}{\lambda_+}\right|^r\right), \qquad r\to\infty,
\end{equation*}
with correlation length
\begin{equation}\label{eq:xi-def}
\xi(\rho,q) = \left[-\log\left|\frac{\lambda_-}{\lambda_+}\right|\right]^{-1}.
\end{equation}
Using $\lambda_+=1+s$, $\lambda_+\lambda_-=z(q-1)$, and $z=s^2(1-\rho)/\rho$ from Appendix~\ref{sec:transfer}, we obtain the closed form
\begin{equation}\label{eq:xi-explicit}
\xi(\rho,q) = \left[ -\log\left| \frac{(q-1)s^2(1-\rho)}{\rho(1+s)^2} \right| \right]^{-1}.
\end{equation}

Two structural features are worth emphasizing. First, the sign of $\lambda_-$ distinguishes the two transport regimes: $\lambda_-/\lambda_+>0$ for $q>1$ produces monotone correlations, reflecting attractive clustering between neighboring particles, while $\lambda_-/\lambda_+<0$ for $q<1$ produces alternating correlations characteristic of an anti-clustered dispersed regime. Second, the point $q=1$ is singular: $\lambda_-=0$ and $\xi=0$, consistent with the Bernoulli product structure of $\pi$ at this point.

The asymptotic behavior of $\xi$ reveals two distinct critical regimes. In the strong-recovery limit $q\to\infty$, one finds $s\sim\sqrt{\rho/(q(1-\rho))}$ and $|\lambda_-/\lambda_+|=1-2s+o(s)$, hence
\begin{equation}\label{eq:xi-large-q}
\xi(\rho,q) \sim \frac12\sqrt{\frac{q(1-\rho)}{\rho}}, \qquad q\to\infty.
\end{equation}
The characteristic size of coherent tandem clusters therefore grows as $q^{1/2}$, which is precisely the scaling of the band widths observed in the right panel of Figure~\ref{fig:spacetime}.

In the opposite limit $q\to0$, the correlation length remains finite away from $\rho=1/2$. At the symmetric density $\rho=1/2$, however, one has exactly $s=q^{-1/2}$, which implies
\[
\left|\frac{\lambda_-}{\lambda_+}\right| = 1-2\sqrt q+o(\sqrt q).
\]
Consequently,
\begin{equation}\label{eq:xi-small-q-half}
\xi(1/2,q) \sim \frac{1}{2\sqrt q}, \qquad q\to0.
\end{equation}

% =========================================================
\section{Conclusion}
\label{sec:discussion}
% =========================================================

We introduced the tandem exclusion process (TEP), a cooperative driven lattice gas motivated by tandem recruitment in ants. Unlike classical exclusion processes, transport in the TEP is carried by dynamically propagating leader-follower pairs rather than by independent particles.

Despite the irreversible nature of the dynamics, the model admits an exact nearest-neighbor Gibbs stationary measure together with explicit formulas for the stationary current and spatial correlations. The cooperative dynamics breaks particle--hole symmetry, shifts the maximal-current regime toward densities larger than one half, and produces strong clustering in the stationary state.

In the strong-recovery regime $q\gg1$, follower recovery becomes effectively instantaneous and tandem clusters propagate coherently through the available vacancies. As a consequence, the stationary current approaches the asymptotic behavior
\[
J(\rho)\sim \alpha\rho,
\]
showing that cooperative coordination can strongly reduce the congestion effects typically associated with exclusion-driven transport.

From a biological perspective, the TEP suggests that direct local coordination alone can substantially enhance collective transport, even without long-range communication mechanisms such as pheromone signaling. The model therefore provides a minimal theoretical framework for understanding how cooperative leader--follower interactions may improve transport efficiency in tandem-running systems.

% =========================================================
\appendix

% =========================================================
\section{Grand-canonical ensemble}
\label{sec:grandcanonical}
% =========================================================

To analyze the thermodynamic limit, it is convenient to introduce the grand-canonical Gibbs measure. For a fugacity $z>0$, define
\begin{equation}\label{eq:GC}
\pi_z(\eta) = \frac{1}{\Xi_N(z,q)} z^{\sum_x\eta_x} q^{H(\eta)}, \qquad \eta\in\{0,1\}^{\TN},
\end{equation}
where
\[
H(\eta)=\sum_{x\in\TN}\eta_x\eta_{x+1},
\]
and
\[
\Xi_N(z,q) = \sum_{\eta\in\{0,1\}^{\TN}} z^{\sum_x\eta_x}q^{H(\eta)}
\]
is the grand partition function.

Conditioning $\pi_z$ on the event
\[
\sum_x\eta_x=K
\]
recovers the canonical stationary measure $\pi_{N,K}$. In the thermodynamic limit
\[
N\to\infty, \qquad K/N\to\rho,
\]
local observables under $\pi_{N,K}$ coincide with those under $\pi_z$, provided the fugacity $z$ is chosen so that
\[
\E_{\pi_z}[\eta_0]=\rho.
\]

The Gibbs weight factorizes over nearest-neighbor bonds:
\[
z^{\sum_x\eta_x}q^{H(\eta)} = \prod_{x\in\TN} z^{\eta_{x+1}} q^{\eta_x\eta_{x+1}}.
\]
This factorization is the starting point of the transfer-matrix representation developed below.

% =========================================================
\section{Transfer-matrix method}
\label{sec:transfer}
% =========================================================

The grand-canonical partition function and stationary correlation functions can be computed explicitly by transfer-matrix methods~\cite{Baxter1982}. Using
\[
H(\eta)=\sum_x\eta_x\eta_{x+1},
\]
the Gibbs weight factorizes over nearest-neighbor bonds as
\[
z^{\sum_x\eta_x}q^{H(\eta)} = \prod_{x\in\TN} T(\eta_x,\eta_{x+1}),
\]
where
\begin{equation}\label{eq:T-matrix}
T = \begin{pmatrix} 1 & z\\ 1 & qz \end{pmatrix}, \quad T(a,b)=z^bq^{ab}, \quad a,b\in\{0,1\}.
\end{equation}
Hence the grand-canonical partition function is
\[
\Xi_N(z,q)=\operatorname{tr}(T^N).
\]

The characteristic polynomial of $T$ is
\begin{equation}\label{eq:chrac_eq_T}
\lambda^2-(1+qz)\lambda+z(q-1)=0,
\end{equation}
with eigenvalues
\begin{equation}\label{eq:eigenvalues}
\lambda_\pm = \frac{(1+qz)\pm\sqrt{(qz-1)^2+4z}}{2}.
\end{equation}
The corresponding right and left eigenvectors associated with the dominant eigenvalue $\lambda_+$ are
\begin{equation}\label{eq:eigenvectors}
v_R = \begin{pmatrix} z\\ \lambda_+-1 \end{pmatrix}, \qquad v_L = \bigl( 1,\lambda_+-1 \bigr),
\end{equation}
with
\[
v_L\!\cdot v_R = z+(\lambda_+-1)^2.
\]

\subsection{Local observables}

Let $f=f(\eta_1,\dots,\eta_r)$ be a local observable depending on $r$ consecutive sites. Expectation with respect to the grand-canonical measure $\pi_z$~\eqref{eq:GC} is
\[
\E_{\pi_z}[f] = \frac1{\Xi_N(z,q)} \sum_{\eta} f(\eta_1,\dots,\eta_r) \prod_{x=1}^N T(\eta_x,\eta_{x+1}).
\]

To isolate the local contribution of $f$, we introduce the insertion matrix
\begin{equation}\label{eq:fhat-def}
\widehat f(a,c) = \sum_{a_2,\dots,a_{r-1}\in\{0,1\}} f(a,a_2,\dots,a_{r-1},c) \prod_{j=1}^{r-1} T(a_j,a_{j+1}),
\end{equation}
where $a_1=a$ and $a_r=c$. For observables depending on a single occupation variable, $f=f(\eta_1)$, we use the convention
\begin{equation}\label{eq:fhat-single}
\widehat f(a,b) = f(b)\,T(a,b).
\end{equation}

Summing over the remaining intermediate variables produces the matrix power $T^{N-r+1}$, and therefore
\begin{equation}\label{eq:observable-trace}
\E_{\pi_z}[f] = \frac{\operatorname{tr}\bigl(\widehat f\,T^{N-r+1}\bigr)}{\operatorname{tr}(T^N)}.
\end{equation}

Since $\lambda_+>|\lambda_-|$, the large-$N$ behavior of the transfer matrix is dominated by the Perron eigenvalue $\lambda_+$:
\[
T^{N-r+1} = \lambda_+^{N-r+1} \frac{v_Rv_L}{v_L\!\cdot v_R} + O(|\lambda_-|^{N-r+1}), \qquad N\to\infty.
\]
Moreover,
\[
\operatorname{tr}(T^N) = \lambda_+^N+O(|\lambda_-|^N).
\]
Substituting these asymptotic expansions into~\eqref{eq:observable-trace}, the leading powers of $\lambda_+$ cancel, yielding
\begin{equation}\label{eq:TM-formula}
\E_{\pi_z}[f] = \frac{v_L\widehat f\,v_R}{\lambda_+^{r-1}(v_L\!\cdot v_R)}.
\end{equation}
In particular, for observables depending on two consecutive sites ($r=2$),
\begin{equation}\label{eq:TM-r2}
\E_{\pi_z}[f] = \frac{v_L\widehat f\,v_R}{\lambda_+(v_L\!\cdot v_R)}.
\end{equation}

\subsection{Density and thermodynamic parametrization}

To compute the density, consider the observable
$
f(\eta_1)=\eta_1.
$
The corresponding insertion matrix is
\[
\widehat f = \begin{pmatrix} 0 & z\\ 0 & qz \end{pmatrix}.
\]
Using~\eqref{eq:eigenvectors} and the eigenvalue relation
\[
z+qz(\lambda_+-1) = \lambda_+(\lambda_+-1),
\]
equation~\eqref{eq:TM-formula} gives
\[
\rho = \E_{\pi_z}[\eta_1] = \frac{(\lambda_+-1)^2}{z+(\lambda_+-1)^2}.
\]

Introduce the parameter
\begin{equation}\label{eq:s-def}
s:=\lambda_+-1>0.
\end{equation}
Then
\begin{equation}\label{eq:rho-s}
\rho = \frac{s^2}{z+s^2}, \qquad z = \frac{s^2(1-\rho)}{\rho}.
\end{equation}

Substituting $\lambda_+=s+1$ into the characteristic equation~\eqref{eq:chrac_eq_T} of $T$ yields
\begin{equation}\label{eq:s-eq-z}
s^2=z+(qz-1)s.
\end{equation}
Combining this with~\eqref{eq:rho-s}, we obtain
\begin{equation}\label{eq:s-eq-app}
q(1-\rho)s^2+(1-2\rho)s-\rho=0,
\end{equation}
which uniquely determines $s=s(\rho,q)$.

At the special point $q=1$,
\[
\det T=z(q-1)=0,
\]
so the transfer matrix becomes rank one. In this case the stationary state reduces to the Bernoulli product measure in the thermodynamic limit.

\subsection{Correlation functions and current}

We now compute the correlation functions entering the stationary current formula. For the head-move pattern
\[
f(\eta_1,\eta_2,\eta_3) = \eta_1\eta_2(1-\eta_3),
\]
corresponding to the local configuration $110$, the insertion matrix is
\[
\widehat f = \begin{pmatrix} 0 & 0\\ qz & 0 \end{pmatrix}.
\]
Applying~\eqref{eq:TM-formula} with $r=3$ gives
\begin{equation}\label{eq:E-head-app}
\E_{\pi_z}[\eta_1\eta_2(1-\eta_3)] = \frac{qz^2\rho}{\lambda_+^2s}.
\end{equation}

Similarly, for the tail-move pattern
\[
f(\eta_1,\eta_2,\eta_3) = \eta_1(1-\eta_2)\eta_3,
\]
corresponding to the configuration $101$, one obtains
\[
\widehat f = \begin{pmatrix} 0 & 0\\ 0 & z \end{pmatrix},
\]
and therefore
\begin{equation}\label{eq:E-tail-app}
\E_{\pi_z}[\eta_1(1-\eta_2)\eta_3] = \frac{z\rho}{\lambda_+^2}.
\end{equation}

Substituting~\eqref{eq:E-head-app} and~\eqref{eq:E-tail-app} into~\eqref{eq:J-as-3pt},
\[
J = \alpha\frac{qz^2\rho}{\lambda_+^2s} + \beta\frac{z\rho}{\lambda_+^2}.
\]
Using $\beta=\alpha q$ and $\lambda_+=s+1$, we obtain
\[
J = \frac{\alpha qz\rho(z+s)}{(s+1)^2s}.
\]

From~\eqref{eq:s-eq-z},
\[
s^2+s=z(1+qs),
\]
which gives
\begin{equation}\label{eq:z-explicit}
z = \frac{s(s+1)}{1+qs}, \qquad z+s = \frac{s(s+2+qs)}{1+qs}.
\end{equation}
Substituting these identities into the previous expression yields
\[
J = \frac{\alpha q\rho s(s+2+qs)}{(s+1)(1+qs)^2},
\]
which proves~\eqref{eq:J-formula}.

Finally, for the nearest-neighbor correlation
\[
f(\eta_1,\eta_2)=\eta_1\eta_2,
\]
the insertion matrix is
\[
\widehat f = \begin{pmatrix} 0 & 0\\ 0 & qz \end{pmatrix}.
\]
Applying~\eqref{eq:TM-r2} gives
\[
\E_{\pi_z}[\eta_1\eta_2] = \frac{qzs^2}{\lambda_+(z+s^2)} = \frac{qz\rho}{\lambda_+}.
\]
Using~\eqref{eq:z-explicit} and $\lambda_+=s+1$,
\begin{equation}\label{eq:E11-app}
\E_{\pi_z}[\eta_1\eta_2] = \frac{q\rho s}{1+qs}.
\end{equation}

% =========================================================

\end{document}